\documentclass[11pt]{article}
\usepackage[utf8]{inputenc}

\usepackage[maxbibnames=99]{biblatex}
\usepackage{amsmath}
\usepackage{amssymb}
\usepackage{amsthm}
\usepackage{tikz}
\usepackage{graphicx}
\usepackage{float}
\usepackage{algorithm}
\usepackage[noend]{algpseudocode}
\usepackage{tabularx}
\usepackage[margin=1in]{geometry}
\usepackage{arydshln}
\usepackage{multirow}
\usepackage{wrapfig}
\usepackage{url}
\usepackage{hyperref}

\makeatletter
\DeclareRobustCommand{\rvdots}{%
  \vbox{
    \baselineskip4\p@\lineskiplimit\z@
    \kern-\p@
    \hbox{.}\hbox{.}\hbox{.}
  }}
\makeatother

\usetikzlibrary{positioning,calc,shapes,arrows}
\usetikzlibrary{backgrounds}
\usetikzlibrary{matrix,shadows,arrows} 
\usetikzlibrary{decorations.pathreplacing,calligraphy}

\def\etal.{et\penalty50\ al.}
\theoremstyle{plain}
\newtheorem{theorem}{Theorem}[section]
\newtheorem{lemma}[theorem]{Lemma}
\newtheorem{proposition}[theorem]{Proposition}

\newtheorem{observation}[theorem]{Observation}

\theoremstyle{definition}
\newtheorem{definition}{Definition}[section]

\theoremstyle{remark}

\theoremstyle{plain}
\newtheorem*{theorem*}{Theorem}

\DeclareMathOperator*{\E}{\mathbb{E}}

\title{Any-Order Online Interval Selection}
\author{Allan Borodin \\ \textsf{bor@cs.toronto.edu} \and Christodoulos Karavasilis \\ \textsf{ckar@cs.toronto.edu}}
\date{}

\begin{document}

\maketitle

\begin{abstract}
We consider the problem of online interval scheduling on a single machine, where intervals arrive online in an order chosen by an adversary, and the algorithm must output a set of non-conflicting intervals. Traditionally in scheduling theory, it is assumed that intervals arrive in order of increasing start times. We drop that assumption and allow for intervals to arrive in any possible order. We call this variant \textit{any-order interval selection} (AOIS). We assume that some online acceptances can be revoked, but a feasible solution must always be maintained. For unweighted intervals and deterministic algorithms, this problem is unbounded. Under the assumption that there are at most $k$ different interval lengths, we give a simple algorithm that achieves a competitive ratio of $2k$ and show that it is optimal amongst deterministic algorithms, and a restricted class of randomized algorithms we call \textit{memoryless}, contributing to an open question by Adler and Azar \cite{adler2003beating}; namely whether a randomized algorithm without memory or with only  ``bounded''  access to history can achieve a constant competitive ratio. We connect our model to the problem of \textit{call control} on the line, and show how the algorithms of Garay et al. \cite{garay1997efficient} can be applied to our setting, resulting in an optimal algorithm for the case of proportional weights. We also discuss the case of intervals with arbitrary weights, and show how to convert the single-length algorithm of Fung et al. \cite{fung2014improved} into a \textit{classify and randomly select} algorithm that achieves a competitive ratio of $2k$. Finally, we consider the case of intervals arriving in a \textit{random order}, and show that for single-lengthed instances, a \textit{one-directional} algorithm (i.e. replacing intervals in one direction), is the only deterministic memoryless algorithm that can possibly benefit from random arrivals.
\end{abstract}

\section{Introduction}
We consider the problem of scheduling intervals online with revoking\footnote{Displacing one or more previously scheduled intervals with a conflicting new interval.}.
Intervals arrive with a fixed start time and fixed end time, and have to be taken right away, or be discarded upon arrival, while no intervals in the solution conflict. The algorithm has to decide which intervals to include in the final schedule, so as to optimize some objective.

In the unweighted case, the goal is to maximize the number of intervals in the final solution. In the weighted case, we want an interval-set of maximum weight. 
Following previous work, we allow some revoking of online decisions, which is often considered even in the conventional start-time-ordered scheduling model. More precisely, if a newly arrived interval conflicts with other intervals already taken by the algorithm, we are able to take the new interval and discard the conflicting intervals. We are able to displace multiple existing intervals at once, although this won't occur in the unweighted case. To avoid confusion, we should note that \textit{preemption}\footnote{In contrast to revoking, preemption in much of the scheduling literature means the pausing of a scheduled job, and resuming it later.} is used in the interval selection literature to mean precisely this revoking of previous decisions we just described. Under this definition, preemption is allowed in our model. When we discard an interval it is final and it cannot be taken again.\\\\
We focus mainly on the unweighted case, where all intervals have the the same weight. We discuss the competitive ratio of the problem in terms of \textit{k}, the number of distinct interval lengths. However our algorithm does not need a priori knowledge of \textit{k}. We show that a simple, deterministic, ``memoryless'' algorithm that only replaces when the new interval is entirely subsumed by an existing one, achieves the optimal competitive ratio in terms of the parameter \textit{k}. We also show that ``memoryless'' randomized algorithms can not do any better.
The main difference between our model and most of the interval selection literature, is allowing intervals to arrive in any order, a strict generalization of the ordered case. Bachmann et al. \cite{bachmann2013online} have studied the any-order input model in the context of ``\textit{t}-intervals'' (we are concerned with $t=1$). They consider randomized algorithms, and don't allow revoking. In that model, they get a lower bound of $\Omega(N)$, with $N$ 
being the number of intervals in a given input instance.
The next most closely related problem is that of call admission \cite{garay1992call} on the line graph, with online intervals corresponding to paths of a given line graph. The connection between call control on the line graph and interval selection has been noted before, but has not been carefully defined. We wish to clarify this connection by explaining the similarities as well as the differences, and how results correspond. We note that the parameter $k \leq N$ (respectively, $k \leq n-1$)   is an obvious   refinement of the number of intervals (respectively, the number  of vertices for call admission on a line graph with $n$ vertices).\\\\
The applications of interval selection problems are plentiful. Some examples are resource allocation, network routing, transportation, and computer wiring. We refer the reader to the surveys by Kolen et al. \cite{kolen2007interval}, and Kovalyov et al. \cite{kovalyov2007fixed} for an overview of results and applications in the area of interval scheduling.\\\\
\textbf{Related Work.} Lipton and Tomkins \cite{lipton1994online} introduced the online interval scheduling problem. In our terminology, they consider the arrival of intervals with increasing start times (ordered), and interval weights that are proportional to the lengths. They don't allow displacement of existing intervals, and give a randomized algorithm with competitive ratio $O((log\Delta)^{1+\epsilon})$, where $\Delta$ is the ratio of the longest to shortest interval.\\\\
In the unweighted case with increasing starting times, Faigle and Nawijn \cite{faigle1995note} give an optimal 1-competitive algorithm that is allowed to revoke previous decisions (replace intervals). In the weighted case with increasing starting times,  Woeginger \cite{woeginger1994line} shows that for general weights, no deterministic algorithm can achieve a constant competitive ratio. Canetti and Irani \cite{canetti1995bounding} extend this and show that even randomized algorithms with revocable decisions cannot achieve a constant ratio for the general weighted case.
For special classes of weight functions based on the length (including proportional weights), Woeginger \cite{woeginger1994line} gives an optimal deterministic algorithm with competitive ratio 4. Seiden \cite{seiden1998randomized} gives a randomized $(2+\sqrt{3})$-competitive algorithm when the weight of an interval is given by a continuous convex function of the length. Epstein and Levin \cite{epstein2008improved} give a $~2.45$-competitive randomized algorithm for weights given by functions of the length that are monotonically decreasing, and they also give an improved $1+\ln(2) \approx 1.693$ upper bound for the weight functions studied by Woeginger \cite{woeginger1994line}. Fung et al. \cite{fung2014improved} currently have the final word on the best upper bounds, giving \textit{barely random} algorithms that achieve a competitive ratio of 2 for all the Woeginger weight functions. These algorithms randomly choose one of two deterministic algorithms at the beginning. More generally, barely random algorithms have access to a small number of deterministic algorithms, and randomly choose one.\\\\
Restricting interval lengths has previously been considered in the literature, e.g. Lipton and Tomkins \cite{lipton1994online} study the case of two possible lengths, and Bachmann et al. \cite{bachmann2013online} consider single and two-length instances. For the related offline problem of throughput maximization, Hyatt-Denesik et al. \cite{hyatt2020approximations} consider $c$ distinct processing times. The special case of single-length jobs has been studied in the job scheduling \cite{sgall1998line,baptiste2000scheduling,chrobak2006note}, sum coloring \cite{borodin2012sum}, and the interval selection literature \cite{fung2012line,miyazawa2004improved}. Woeginger \cite{woeginger1994line} also points out how his results can be extended to the case of equal lengths and arbitrary weights. Miyazawa and Erlebach \cite{miyazawa2004improved} point out the equivalency between fixed length (w.l.o.g. unit) instances, and 
proper interval instances, i.e. instances where no interval is contained within another. This is because of a result by Bogart and West \cite{bogart1999short}, showing the equivalency of the corresponding interval graphs in the offline setting.\\\\
There has also been some work on multiple identical machines. For the case of equal-length, arbitrary-weight intervals, Fung et al. \cite{fung2012line} give an algorithm that is 2-competitive when $m$, the number of machines, is even, and $(2+\frac{2}{2m-1})$ when $m$ is odd. Yu \& Jacobson \cite{yu2018online} consider C-benevolent (weight function is convex increasing) jobs and get an algorithm that is 2-competitive when $m$ is even, and $(2+\frac{2}{m})$-competitive when $m$ is odd.\\\\
In the problem of call control, a graph is given, and requests that correspond to pairs of nodes of the graph arrive online. The goal is to accept as many requests as possible, with the final set consisting of disjoint paths. When the underlying graph is a line, this problem is closely related to ours. For call control on the line, Garay et al. \cite{garay1997efficient} give optimal deterministic algorithms. In the unweighted case, they achieve a $O(\log(n))$ competitive ratio, where $n$ is the number of the vertices of the graph. In the case of proportional weights (weight is equal to the length of the path), they give an optimal algorithm that is $(\sqrt{5}+2)\approx 4.23$-competitive (its optimality was shown by Furst and Tomkins \cite{tomkins1995lower}). Adler and Azar \cite{adler2003beating} use randomization to overcome the $\log(n)$ lower bound, and give a 16-competitive algorithm. Emek et al. \cite{emek2016space} study interval selection in the streaming model, and show how to modify their streaming algorithm to work online, achieving a competitive ratio of 6, improving upon the 16-competitive algorithm of Adler and Azar. It is noteworthy that the Adler and Azar algorithm uses memory proportional to the entire input sequence.
In contrast, the Emek et al. algorithm only uses memory that is within a constant factor of a current OPT solution.  It is still an open question if a randomized algorithm using only constant bounded memory can get a constant ratio in the unweighted case. We show that  for a strict, but natural definition of memoryless randomized algorithms, a constant ratio cannot be obtained. The algorithms presented in this paper, along with the optimal algorithms by Garay et al. \cite{garay1997efficient} and Woeginger \cite{woeginger1994line}, fall under our definition of memoryless. It is worth noting that similar notions of memoryless algorithms, and comparison between randomized memoryless and deterministic, have appeared in the k-server and caching literature \cite{coester2019online,koutsoupias2009k,Kleinberg94,raghavan1989memory}. We would note that barely random algorithms as described earlier (i.e. algorithms that initially generate some random bits, which are used in every online step), are not memoryless but usually satisfy bounded memory.  The algorithms by Fung et al. \cite{fung2014improved} are an example of this. More generally, this use of initial random bits are the \textit{classify and randomly select} algorithms\footnote{Barely random algorithms can be thought of as a special case of the classify and randomly select paradigm.} (e.g. Lipton and Tomkins \cite{lipton1994online} and Awerbuch et al. \cite{awerbuch1994competitive}). It's important to note that such algorithms may require prior knowledge of bounds on lengths of intervals. In appendix \ref{app:A} we discuss our meaning of memoryless and bounded memory online algorithms, and the relation to randomness, advice,  and the Adler and Azar question. \\\\
The problem of admission control has also been studied under the model of minimizing rejections \cite{blum2001admission,alon2005admission} instead of maximizing acceptances. An alternative input model for interval selection is that of arriving conflicts \cite{halldorsson2013online} instead of single intervals, with the algorithm being able to choose at most one item from each conflict. We also note that, an instance of interval selection can be represented as an interval graph, with intervals corresponding to vertices, and edges denoting a conflict between two intervals. Generally, interval graphs reveal much less about the instance compared to receiving the actual intervals. In the interval graph representation, arriving vertices may have an adjacency list only in relation to already arrived vertices, or they may show adjacency to future vertices as well.\\\\
\textbf{Our results.} For the unweighted adversarial case, we know that no deterministic algorithm is bounded (follows from \cite{garay1997efficient}). Assuming there are at most $k$ different lengths, we show how a simple greedy algorithm achieves a competitive ratio of $2k$. We also give a matching lower bound that holds for all deterministic algorithms, as well as ``memoryless'' randomized algorithms. We note that an instance with $k$ different lengths can have a nesting depth of at most $k-1$. Alternatively, we can state our results in terms of $d$, the nesting depth (see figure \ref{fig:conflicts}), noting that $d \leq (k-1)$. This implies that our $2k$ bounds can be restated as $2(d+1)$. We also show how to extend  the classify and randomly select paradigm used by Fung et al \cite{fung2014improved} to obtain a randomized algorithm that is $2k$-competitive for the case of arbitrary weights and $k$ different interval lengths. It's worth noting that Canetti and Irani \cite{canetti1995bounding} give a $\Omega (\sqrt{k})$ lower bound for randomized algorithms and arbitrary weights.\\\\
We show how the problem of call control on the line \cite{garay1997efficient} relates to interval selection, and in particular how their $\log n$-competitive algorithm for the unweighted case and their  $(2+\sqrt{5})$-competitive algorithm for proportional weights carries over to interval selection. Lastly, we consider deterministic memoryless algorithms for the problem of any-order, unweighted, single-lengthed (i.e. unit) intervals with random order arrivals. We show that the only deterministic memoryless algorithm that can possibly benefit from random arrivals is one-directional, only replacing intervals if they overlap in that particular direction. \\\\
\textit{Organization of the paper.}
Section 2 has some definitions to clarify the model. Section 3 has our upper and lower bounds in the adversarial case, the connection to call control, and the application of the proportional weights algorithm to our model. Section 4 discusses arbitrary weights. Section 5 is about interval selection in the  random order model. We end with some conclusions and open problems.

\section{Preliminaries}

Our model consists of intervals arriving on the real line. An interval $I_{i}$ is specified by a starting point $s_{i}$, and an end point $f_{i}$, with $s_{i} < f_{i}$. It occupies space $[s_{i},f_{i})$ on the line, and the conventional notions of intersection, disjointness, and containment apply. This allows two adjacent intervals $[s_{1},f)$ and $[f,f_{2})$ to not conflict, although our results would apply even if we considered closed intervals $[s_{i},f_{i}]$ with $[s_{1},f]$ and $[f,f_{2}]$ conflicting. There are two main ways two intervals can conflict, and they are shown in figure \ref{fig:conflicts}.\\\\
We use the notion of competitive ratio to measure the performance of our online algorithms. Given an algorithm $A$, let ALG denote the objective value of the solution achieved by the algorithm, and let OPT denote the optimal value achieved by an offline algorithm. The competitive ratio of $A$ is defined as follows: $CR(A)=\frac{OPT}{ALG} \geq 1$. We should note that we can repeat disjoint copies of our nemesis sequences, and get the corresponding tight lower bounds. As a result, we can omit the standard additive term in our definition of competitive ratio. We will sometimes abuse notation and use ALG and OPT to denote the sets of intervals maintained by the algorithm at some given point, and the set of intervals of an optimal solution respectively. In the case of deterministic algorithms and random arrival of intervals, the performance of an algorithm is a random variable, and the competitive ratios hold w.h.p. (definition of competitive ratio remains unchanged). The algorithm we present in the case of arbitrary weights is randomized, and its expected competitive ratio is defined as $CR(A)=\frac{OPT}{\E[ALG]}$.
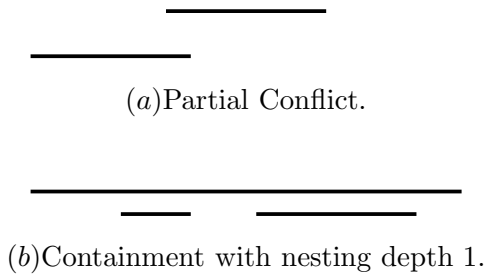
\begin{figure}[H]
	\centering
	
	\begin{tikzpicture}[scale=0.6]

	\node[draw=none] (I1a) at (-10,0) {$ $};
	\node[draw=none] (I1b) at (-6,0) {$ $};
	\draw[line width=0.5mm] (I1a) -- (I1b);


	\node[draw=none] (I2a) at (-7,1) {$ $};
	\node[draw=none] (I2b) at (-3,1) {$ $};
	\draw[line width=0.5mm] (I2a) -- (I2b);
	
	\node at (-5,-1) {$(a) \text{Partial Conflict.}$};
	
	\node[draw=none] (I3a) at (-8,-3.5) {$ $};
	\node[draw=none] (I3b) at (-6,-3.5) {$ $};
	\draw[line width=0.5mm] (I3a) -- (I3b);

 \node[draw=none] (Ima) at (-5,-3.5) {$ $};
	\node[draw=none] (Imb) at (-1,-3.5) {$ $};
	\draw[line width=0.5mm] (Ima) -- (Imb);


	\node[draw=none] (I4a) at (-10,-3) {$ $};
	\node[draw=none] (I4b) at (-0,-3) {$ $};
	\draw[line width=0.5mm] (I4a) -- (I4b);
	
	\node at (-5,-4.5) {$(b) \text{Containment with nesting depth 1}.$};

	\end{tikzpicture} 
	\caption{Types of conflicts.}\label{fig:conflicts}
\end{figure}
We sometimes refer to a \textit{chain} of intervals (figure \ref{fig:chain-ex}). This is a set of intervals where each interval partially conflicts with exactly two other intervals, except for the two end intervals that partially conflict with only one.
\begin{figure}[H]
	\centering
	
	\begin{tikzpicture}[scale=0.6]

	\node[draw=none] (I1a) at (-18,0) {$ $};
	\node[draw=none] (I1b) at (-14,0) {$ $};
	\draw[line width=0.5mm] (I1a) -- (I1b);


	\node[draw=none] (I2a) at (-15,1) {$ $};
	\node[draw=none] (I2b) at (-11,1) {$ $};
	\draw[line width=0.5mm] (I2a) -- (I2b);

    \node[draw=none] (I3a) at (-12,0) {$ $};
	\node[draw=none] (I3b) at (-8,0) {$ $};
	\draw[line width=0.5mm] (I3a) -- (I3b);

 \node[draw=none] (I4a) at (-9,1) {$ $};
	\node[draw=none] (I4b) at (-5,1) {$ $};
	\draw[line width=0.5mm] (I4a) -- (I4b);

 \node[draw=none] (I5a) at (-6,0) {$ $};
	\node[draw=none] (I5b) at (-2,0) {$ $};
	\draw[line width=0.5mm] (I5a) -- (I5b);

	\end{tikzpicture} 
	\caption{Interval chain.}\label{fig:chain-ex}
\end{figure}
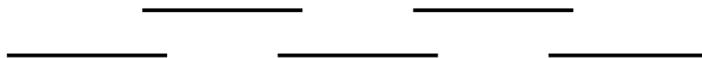

\section{Adversarial Order}
\subsection{Unweighted}
In this section, we assume an adversary chooses the instance configuration, along with the arrival order of all intervals. Lemma \ref{lem:need-rev} shows that revocable decisions are necessary even in the case of two different lengths. Algorithm \ref{alg:SUB} is the greedy algorithm that achieves the optimal competitive ratio of $2k$ in the unweighted case, and it works as follows: On the arrival of a new interval, take it if there's no conflict. If there's a conflict, take the new interval only if it is properly contained inside an existing interval.

\begin{algorithm}
\caption{}\label{alg:SUB}
\begin{algorithmic}

\State On the arrival of $I$:
\State $I_{s} \gets $ Set of intervals currently in the solution conflicting with $I$

\For{$I' \in I_{s}$}
\If{$I \subset I'$}
    \State Take $I$ and discard $I'$
    \State return
\EndIf
\EndFor
\State Discard $I$

\end{algorithmic}
\end{algorithm}

\begin{lemma}
\label{lem:need-rev}
The problem of any-order unweighted interval scheduling with two different lengths and irrevocable decisions is unbounded.
\end{lemma}
\begin{proof}
Consider two possible interval lengths of 1 and $K$. Let an interval of length $K$ arrive first. W.l.o.g. the algorithm takes it. Then $K$ 1-length intervals arrive next, all of them overlapping with first $K$-length interval. The algorithm cannot take any of the 1-length intervals, achieving a competitive ratio of $\frac{1}{K}$.
\end{proof}
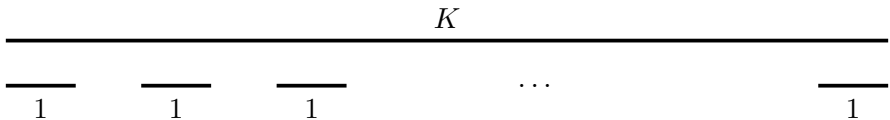
\begin{figure}[H]
	\centering
	\begin{tikzpicture}[scale=0.6]
	
	\node at (0,0.5) {$K$};
	\node[draw=none] (I1a) at (-10,0) {$ $};
	\node[draw=none] (I1b) at (10,0) {$ $};
	\draw[line width=0.5mm] (I1a) -- (I1b);
	
    \node at (-9,-1.5) {$1$};
	\node[draw=none] (I2a) at (-10,-1) {$ $};
	\node[draw=none] (I2b) at (-8,-1) {$ $};
	\draw[line width=0.5mm] (I2a) -- (I2b);
	
	\node at (-6,-1.5) {$1$};
	\node[draw=none] (I3a) at (-7,-1) {$ $};
	\node[draw=none] (I3b) at (-5,-1) {$ $};
	\draw[line width=0.5mm] (I3a) -- (I3b);
	
	\node at (-3,-1.5) {$1$};
	\node[draw=none] (I4a) at (-4,-1) {$ $};
	\node[draw=none] (I4b) at (-2,-1) {$ $};
	\draw[line width=0.5mm] (I4a) -- (I4b);
	
	\node[draw=none] (d1) at (2,-1) {$\dots$};
	
	\node at (9,-1.5) {$1$};
	\node[draw=none] (I5a) at (8,-1) {$ $};
	\node[draw=none] (I5b) at (10,-1) {$ $};
	\draw[line width=0.5mm] (I5a) -- (I5b);

	\end{tikzpicture} 
	\caption{Unweighted instance with two different lengths.
}\label{fig:need-revoke}
\end{figure}

\begin{theorem}\label{theo:pos-2k}
Algorithm 1 achieves a competitive ratio of $2k$ for the problem of any-order unweighted interval scheduling with $k$ different lengths.
\end{theorem}
\begin{proof}
We define a mapping of intervals $f: OPT \longrightarrow ALG$, where every interval in ALG has at most $2k$ intervals in OPT mapped to it. Because intervals taken by the algorithm might be replaced during the execution, the mapping $f$ might be redefined multiple times. What follows is the way optimal intervals $I \in OPT$ are charged, as soon as they arrive, to intervals $I' \in ALG$. There are four cases of interest:\\\\
\textit{Case 1}: The newly arrived optimal interval is taken by the algorithm.\\
This can happen either because this interval did not conflict with any other intervals taken by the algorithm, or because it was entirely subsumed by a larger interval in ALG, in which case the algorithm would have replaced the large interval with the new small one. In this case, this optimal interval is mapped onto itself.\\\\
\textit{Case 2}: The newly arrived optimal interval partially conflicts with one interval currently in ALG.
In this case, this optimal interval is charged to the interval it conflicts with.\\\\
\textit{Case 3}: The newly arrived optimal interval partially conflicts with two intervals currently in ALG.
In this case, this optimal interval can be charged to any of these two intervals arbitrarily. We may assume it's always charged to the interval it conflicts with on the right. Notice also, that a newly arrived interval, cannot partially conflict with more than two intervals in ALG.\\\\
\textit{Case 4}: The newly arrived optimal interval subsumes an interval currently in ALG. W.l.o.g. we can assume this never happens. Any such optimal solution $OPT$ can be turned into an optimal solution $OPT'$, with the smaller interval in place of the larger one. We can restrict ourselves to only look at optimal solutions where no such transformation can take place. This case also encapsulates the case of an optimal interval perfectly coinciding with an interval taken by the algorithm.\\\\
An interval ($I_{l}$) taken by the algorithm can later be replaced, if a smaller one ($I_{s}$) comes along and is subsumed by it. When this happens, all intervals in $OPT$ charged to $I_{l}$ up to that point, will be transferred and charged to $I_{s}$. As a result, there are two ways an interval taken by the algorithm can be charged intervals in $OPT$. The first way is when an interval $I \in OPT$ is directly charged to an interval $I' \in ALG$ when $I$ arrives (Cases 1-4). This will be referred to as \textit{direct charging}. The second way is when a new interval, $I_{n}$, arrives, and replaces an existing interval $I_{e}$, in which case all optimal intervals previously charged to $I_{e}$, will now be charged to $I_{n}$. This will be referred to as \textit{transfer charging}.\\\\

\begin{proposition}
\label{prop:direct}
An interval taken by the algorithm (even temporarily), can be charged by at most two optimal intervals through \textit{direct charging}.
\end{proposition}
To see why this true, we consider the three main cases of direct charging explained earlier. In \textit{Case} 1, the optimal interval is taken by the algorithm and is charged to itself. Because no other optimal interval can conflict with it, we know this interval will never be directly charged again.\\
In \textit{Cases} 2 and 3, direct charging happens because of the optimal interval partially conflicting with one or two intervals currently taken by the algorithm. Because an interval taken by the algorithm can partially conflict with at most two optimal intervals (one on each side), it can be charged twice at most.\\
\begin{proposition}
\label{prop:transfer}
An interval taken by the algorithm can be charged at most $2k-2$ optimal intervals through \textit{transfer charging.}
\end{proposition}
Consider a sequence of interval replacements by the algorithm, where all optimal intervals charged to an interval in the sequence, are passed down to the next interval in the sequence. The last interval in that sequence will have accumulated all the optimal intervals charged to the previous intervals in that sequence. Because we consider $k$ different lengths, such a sequence can have up to $k$ intervals, participating in $k-1$ transfer charging events. We also know that every interval in that sequence can be charged at most two optimal intervals through direct charging (\textit{Proposition \ref{prop:direct}}) before being replaced.
Consequently, assuming two additional charges are added to each interval in that sequence, the last (smallest) interval will be charged $2(k-1)$ optimal intervals through transfer charging.\\\\
We have described a process, during which every optimal interval is charged to an interval in ALG.
By Propositions \ref{prop:direct} \& \ref{prop:transfer}, we know that an interval in ALG, can be charged by $2k$ intervals in OPT at most. Therefore, our algorithm has a competitive ratio of $2k$ for the problem of unweighted interval selection with revocable decisions and $k$ different possible interval lengths. This ratio is tight for this algorithm and an example instance for $k=2$ is shown in Figure 4. $I_{1}$ is directly charged by $I_{2}$ and $I_{3}$, transfers charges to $I_{4}$, which in turn is directly charged another two times by $I_{5}$ and $I_{6}$.
\end{proof}

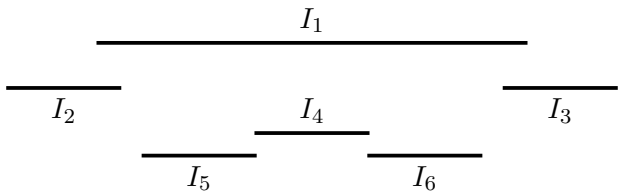
\begin{figure}[H]
	\centering
	
	\begin{tikzpicture}[scale=0.6]

	\node at (0,0.5) {$I_{1}$};
	\node[draw=none] (I1a) at (-5,0) {$ $};
	\node[draw=none] (I1b) at (5,0) {$ $};
	\draw[line width=0.5mm] (I1a) -- (I1b);
	
	\node at (-5.5,-1.5) {$I_{2}$};
	\node[draw=none] (I11a) at (-7,-1) {$ $};
	\node[draw=none] (I11b) at (-4,-1) {$ $};
	\draw[line width=0.5mm] (I11a) -- (I11b);
	
	\node at (5.5,-1.5) {$I_{3}$};
	\node[draw=none] (I21a) at (4,-1) {$ $};
	\node[draw=none] (I21b) at (7,-1) {$ $};
	\draw[line width=0.5mm] (I21a) -- (I21b);

	\node[draw=none] (I4a) at (-1.5,-2) {$ $};
	\node[draw=none] (I4b) at (1.5,-2) {$ $};
	\node at (0,-1.5) {$I_{4}$};
	\draw[line width=0.5mm] (I4a) -- (I4b);
	
	\node[draw=none] (I5a5) at (-4,-2.5) {$ $};
	\node[draw=none] (I5b5) at (-1,-2.5) {$ $};
	\node at (-2.5,-3) {$I_{5}$};
	\draw[line width=0.5mm] (I5a5) -- (I5b5);
	
	\node[draw=none] (I6a) at (1,-2.5) {$ $};
	\node[draw=none] (I6b) at (4,-2.5) {$ $};
	\node at (2.5,-3) {$I_{6}$};
	\draw[line width=0.5mm] (I6a) -- (I6b);
	
	\end{tikzpicture} 
	\caption{4-competitive tight example for Algorithm 1. Interval subscripts corresponds to the arrival order.}\label{fig:4-tight}
\end{figure}
We now provide a matching lower bound, showing that no deterministic algorithm can do better.
\begin{theorem}
\label{theo:neg-2k}
No deterministic algorithm can achieve a competitive ratio better than $2k$ for the problem of unweighted interval selection with revocable decisions and $k$ different lengths.
\end{theorem}
\begin{proof}
At any point during the execution, the algorithm will have exactly one interval in its solution, while the size of the optimal solution will keep growing. We begin by describing how the main component of the instance is constructed, using intervals of the same length. First, the adversary must decide on an overlap amount $v$, which can be arbitrary. All partially conflicting intervals will overlap by exactly this amount. Consider now the instance of figure \ref{fig:base-ADV}. Intervals $I_{1}$ and $I_{2}$ arrive first in that order. If $I_{1}$ is taken by the algorithm and is then replaced by $I_{2}$, then $I_{4}$ arrives. If $I_{1}$ was taken by the algorithm but was not replaced by $I_{2}$, then $I_{3}$ would arrive. Because this case is symmetrical, we only consider the former case of $I_{2}$ replacing $I_{1}$. What happens is that this chain keeps growing in the same direction, until the algorithm decides to stop replacing. When that happens, we look at the last three intervals of the chain. For example, when $I_{4}$ arrived, if the algorithm chose to not select $I_{4}$ and instead maintain $I_{2}$, we stop growing the chain and consider the intervals $(I_{1}, I_{2}, I_{4})$. If the algorithm never stops replacing, it will end up with $I_{5}$ in its solution. Although it's not necessary, if the algorithms seems to always be replacing as the chain is growing, the adversary is able to abuse this as much as they want. In all cases, there exists an optimal solution of at least two intervals, with neither of them being the one taken by the algorithm. Note also that this construction requires at most four intervals of length $L$, occupying space at most $(4L - 3v)$ in total.\\
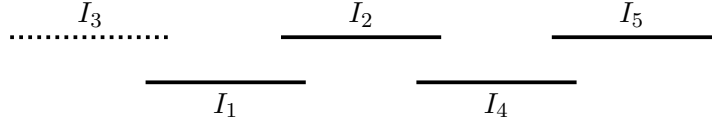
\begin{figure}
	\centering
	
	\begin{tikzpicture}[scale=0.6]

	\node at (-8,-0.5) {$I_{1}$};
	\node[draw=none] (I1a) at (-10,0) {$ $};
	\node[draw=none] (I1b) at (-6,0) {$ $};
	\draw[line width=0.5mm] (I1a) -- (I1b);


	\node[draw=none] (I2a) at (-7,1) {$ $};
	\node[draw=none] (I2b) at (-3,1) {$ $};
	\node at (-5,1.5) {$I_{2}$};
	\draw[line width=0.5mm] (I2a) -- (I2b);
	
	\node[draw=none] (I3a) at (-13,1) {$ $};
	\node[draw=none] (I3b) at (-9,1) {$ $};
	\node at (-11,1.5) {$I_{3}$};
	\draw[dotted,line width=0.5mm] (I3a) -- (I3b);
	
	\node[draw=none] (I4a) at (-4,0) {$ $};
	\node[draw=none] (I4b) at (0,0) {$ $};
	\node at (-2,-0.5) {$I_{4}$};
	\draw[line width=0.5mm] (I4a) -- (I4b);
	
	\node[draw=none] (I5a) at (-1,1) {$ $};
	\node[draw=none] (I5b) at (3,1) {$ $};
	\node at (1,1.5) {$I_{5}$};
	\draw[line width=0.5mm] (I5a) -- (I5b);

	\end{tikzpicture} 
	\caption{Base adversarial construction}\label{fig:base-ADV}
\end{figure}
A small detail is that w.l.o.g. we can assume $I_{1}$ is always taken by the algorithm when it first arrives. Because this construction will take place a number of times during the execution, when the algorithm will already have an interval in its solution, it's useful to consider the case when $I_{1}$ is not taken by the algorithm. In this case, we start growing the chain regardless. If $I_{2}$ or $I_{4}$ are taken by the algorithm, we treat it similarly to when $I_{1}$ was taken and the algorithm kept replacing. If the algorithm hasn't taken any interval even after $I_{4}$ has arrived, the chain stops growing and we consider the intervals $(I_{1}, I_{2}, I_{4})$.\\\\
Let $I_{alg}$ be the interval taken by the algorithm (or $I_{2}$ if no intervals were taken). All remaining intervals to arrive will be subsumed by $I_{alg}$, and thus will not conflict with the two  neighboring intervals taken by $OPT$. Assuming $I_{alg}$ conflicts with one interval on the left and one on the right, that leaves space of length ($L - 2v$) for all remaining intervals. Inside that space, the exact same construction described will take place, only when the algorithm takes a new interval, it implies $I_{alg}$ is replaced. This can be thought of as going a level deeper, and using a sufficiently smaller interval length. More precisely, if $L'$ is the new (smaller) length that will be used, it must hold that $L' \leq \frac{L+v}{4}$.\\\\
After each such construction is completed, the size of the optimal solution grows by at least 2. Because there are at most $k$ different lengths, this can be repeated at most $k$ times. Finally, because the algorithm only ever keeps a single interval in its solution, it will achieve a competitive ratio of $2k$.
\end{proof}
We now extend Theorem \ref{theo:neg-2k} and show that the the $2k$ lower bound also holds for a class of randomized algorithms we call \textit{memoryless}. Intuitively, memoryless algorithms decide on taking or discarding the newly arrived interval, only by looking at the new interval, and all the intervals currently in the solution, using no information from previous online rounds. Although not randomized, it's worth noting that Algorithm \ref{alg:SUB}, along with the optimal deterministic algorithms for call control \cite{garay1997efficient}, are memoryless.
\begin{definition}[\textit{Memoryless randomized algorithm}]
We call a randomized algorithm memoryless, if a newly arrived interval $I_{new}$ is taken with probability $F(I_{new},S)$, where $S=\{I_{1},I_{2},...\}$ is the set of intervals currently in the solution, and each interval is a tuple of the form $(s_{i},f_{i})$.
\label{def:mem-rand}
\end{definition}
Notice that definition \ref{def:mem-rand} only allows us to make use of random bits of this current step, and it does not allow access to random bits from previous rounds. In particular, this definition does not capture barely random algorithms (as mentioned in the introduction), or algorithms that fall under the \textit{classify and randomly select} paradigm.\\

\begin{theorem}
No memoryless randomized algorithm can achieve a competitive ratio better than $2k$ for the
problem of unweighted interval selection with revocable decisions and k different lengths. More specifically, for all $p \in (0,1]$, there exists an $\epsilon_{p} >0$, such that the competitive ratio is greater than $2k-\epsilon_{p}$ with probability $p$.
\end{theorem}
\begin{proof}
The proof is very similar to the proof of Theorem \ref{theo:neg-2k}. The instance has the same structure as the one described in the proof of Theorem \ref{theo:neg-2k}, with the difference that whenever a new interval is taken with probability $p>0$, the adversary will have to add as many copies of that interval as necessary, so that it's taken w.h.p. Figure \ref{fig:no-mem-rand} shows an example of multiple copies of a new interval, ensuring that a replacement happens w.h.p.
\end{proof}
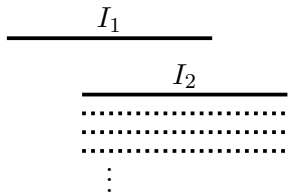
\begin{figure}[H]
	\centering
	
	\begin{tikzpicture}[scale=0.5]

	\node at (3,1) {$I_{1}$};
	\node[draw=none] (I1a) at (0,0.5) {$ $};
	\node[draw=none] (I1b) at (6,0.5) {$ $};
	\draw[line width=0.5mm] (I1a) -- (I1b);

    \node at (5,-0.5) {$I_{2}$};
	\node[draw=none] (I2a) at (2,-1) {$ $};
	\node[draw=none] (I2b) at (8,-1) {$ $};
	\draw[line width=0.5mm] (I2a) -- (I2b);
	\node[draw=none] (I2aa) at (2,-1.5) {$ $};
	\node[draw=none] (I2bb) at (8,-1.5) {$ $};
	\draw[dotted,line width=0.5mm] (I2aa) -- (I2bb);
	\node[draw=none] (2I2aa) at (2,-2) {$ $};
	\node[draw=none] (2I2bb) at (8,-2) {$ $};
	\draw[dotted,line width=0.5mm] (2I2aa) -- (2I2bb);
	\node[draw=none] (3I2aa) at (2,-2.5) {$ $};
	\node[draw=none] (3I2bb) at (8,-2.5) {$ $};
	\draw[dotted,line width=0.5mm] (3I2aa) -- (3I2bb);
	
	\node[draw=none] (d1) at (3,-3.3) {$\rvdots$};

	\end{tikzpicture} 
	\caption{Replacing $I_{1}$ w.h.p when $F(I_{2},\{I_{1}\}) > 0.$
}\label{fig:no-mem-rand}
\end{figure}
It is worth mentioning that similar to how we extend our lower bound to hold for memoryless randomized algorithms, one can extend the $\log(n)$ lower bound for call control \cite{garay1997efficient} to also hold for memoryless randomized algorithms. We also prove the following lower bound on all randomized algorithms and instances with no proper inclusions, capturing the case of $k=1$.
\begin{lemma}
    No randomized algorithm can achieve a competitive ratio better than $\frac{4}{3}$ for the case of unweighted instances with no proper inclusions.
\end{lemma}
\begin{proof}
    Consider two different input sequences, $S_{1}$ and $S_{2}$, each consisting of a chain of intervals. The first two intervals $(I_{1},I_{2})$ are the same in both sequences and partially conflict, with $f_{2} > f_{1}$. $S_{1}$ has the third interval in the sequence partially conflict with $I_{2}$, whereas $S_{2}$ has the third interval conflict with $I_{1}$.
    The adversary chooses one of the two sequences with probability $\frac{1}{2}$. Regardless of what the algorithm does after the arrival of $I_{2}$, it ends up with a single interval in its solution with probability $\frac{1}{2}$, whereas the optimal solution is always of size $2$. This leads to an expected solution size of $\frac{3}{2}$, and a competitive ratio of $\frac{4}{3}$.    
\end{proof}

\subsection{Connection to Call Control}
In this section we relate our results to those of call control on the line.
In the problem of call control, we are given a graph $G(V,E)$, and requests (intervals) correspond to paths on the graph. We note that in the call control literature, it is assumed that requests can come in any order. The length of a request is defined as the length of the corresponding path, and a valid solution is a set of edge-disjoint paths. The objective is to maximize the number of accepted paths. The special case of line graphs is of most relevance to us, with Garay et al. \cite{garay1997efficient} giving an optimal $\log(n)$-competitive algorithm, where $n=|V|$. Their algorithm is similar to ours, with one important additional replacement rule. If the new interval's length is less than half the length of the shortest interval it's conflicting with, the new interval is taken by displacing whatever is necessary. Notice how by fixing each edge to be of the same Euclidean length, that notion of length matches ours, without modifying the instance in any meaningful way.\\\\
An apparent difference between our model and that of call control, is that in the latter, we're initially given the graph. It is not clear how an algorithm that uses that information would operate in our setting. 
The optimal algorithm of \cite{garay1997efficient} does not use that information. We can use such an algorithm on an instance of AOIS, using our definition of length. To see how the $\log(n)$ upper bound would still apply, we describe a way to add vertices on the line after the entire AOIS instance has been revealed, such that we can view it as an instance of call control. There isn't just one way to add those vertices \footnote{Notice how we can take an instance of call control on $n$ points, and repeatedly add $n-1$ new points, one new point between every pair of consecutive points, while keeping it a valid instance.}, and ideally we're interested in the minimum number of vertices for a ``valid'' call control instance to be formed\footnote{The total number of vertices we add to view it as an instance of call control has no impact on the quality (i.e., the number of intervals accepted) of the solution. The $\log(n)$ is in terms of the minimum n.}. There are two requirements for the point-adding construction to be considered valid: (a) the points must be equally spaced, so that the two definitions of length in our model and call control match, and (b) every start and end point of an interval must coincide with a point. If all interval start and end points are rational numbers, we can multiply them by a common denominator, and use integer points. If there exist interval start/end points that are irrational, we can approximate the call control instance by adding sufficiently many points that are sufficiently close to each other, so that our comparison of two lengths in the AOIS instance, gives the same result as in the final call control instance. This is a technical issue we leave as an open problem.
Such a construction allows us to view the AOIS instance as an instance of call control. If the resulting graph was given to us a priori and we applied the call control algorithm to it, the final solution would be the same. This also allows us to use the $(\sqrt{5}+2)$-competitive algorithm by Garay et al. \cite{garay1997efficient} in the case of interval weights proportional to their length.\\\\
Applying an algorithm for AOIS on a call control instance is more straightforward. W.l.o.g. we can fix edge lengths to be unit lengths and directly apply the algorithm, while achieving the same competitive ratio.\\\\
To see how our $2k$ lower bound applies to call control on the line, notice that being allowed $k$ different lengths, there need to be enough points to allow the adversary to fit in the instance of Theorem \ref{theo:neg-2k}. Given that construction, we can compute a lower bound on the number of points (vertices) required. The base (shortest) intervals need two points each. Each level (base one included) has at most four intervals, all of the same length. Because the points are equally spaced, same-length intervals cover the same number of points. If an interval on level $j$ covers $x$ points, an interval of the upper (longer) level $j+1$ covers at least $4x$ points. This results in a lower bound of $2^{2k+1}$ points. Another way to view this, is that if the number of different lengths is sufficiently small (compared to $n$), the $2k$ lower bound applies, and our algorithm becomes optimal for call control on the line. In particular, this argument does not contradict the known $log(n)$ lower bound.\\\\
Whether the algorithm by Garay et al. can be forced to a competitive ratio worse than $2k$ for AOIS is open. When deciding on which algorithm to use, potential knowledge about the instance structure may be of help. This becomes apparent from the two following observations.
\begin{observation}
There exists an instance where the algorithm by Garay et al. \cite{garay1997efficient} achieves a competitive ratio of $(2k-2)$, whereas algorithm \ref{alg:SUB} gets 1.
\end{observation}
This instance is essentially two long chain-like structures that meet in the middle, and it is depicted in figure \ref{fig:garay-2k}. A single side of the construction is better shown in figure \ref{fig:garay-2k-helper}.  The arrival sequence is $L_{1},R_{1},L_{2},R_{2},L'_{2},R'_{2},...,L_{k-1},R_{k-1},L'_{k-1},R'_{k-1},M$. We have that $|L_{i}|=|L'_{i}|=|R_{i}|=|R'_{i}|$, and $|L_{i}|>2|L_{i+1}|$. The algorithm in \cite{garay1997efficient} would always take the next interval on the chains $L_{1},L_{2},...$ and $R_{1},R_{2},...$, ending up with only $M$, which would displace both intervals in the solution $\{L_{k-1},R_{k-1}\}$. The optimal solution $\{L_{1},R_{1},L'_{2},R'_{2},...,L'_{k-1},R'_{k-1}\}$ is obtained by our algorithm, and it is of size $(2k-2)$.
\begin{figure}[H]
	\centering
	
	\begin{tikzpicture}[scale=0.4]
	
	\node at (-12,-0.5) {$L_{1}$};
	\node[draw=none] (L1a) at (-16,-1) {$ $};
	\node[draw=none] (L1b) at (-8,-1) {$ $};
	\draw[line width=0.5mm] (L1a) -- (L1b);

 \node at (-7,-1) {$L_{2}$};
	\node[draw=none] (L1a) at (-9,-1.5) {$ $};
	\node[draw=none] (L1b) at (-5,-1.5) {$ $};
	\draw[line width=0.5mm] (L1a) -- (L1b);

    \node at (-6,-2.5) {$L'_{2}$};
	\node[draw=none] (L11a) at (-8,-2) {$ $};
	\node[draw=none] (L11b) at (-4,-2) {$ $};
	\draw[line width=0.5mm] (L11a) -- (L11b);

 \node[draw=none] (d1) at (-4,-1.5) {$\dots$};
 \node[draw=none] (d2) at (4,-1.5) {$\dots$};
	
	\node at (11.5,-0.5) {$R_{1}$};
	\node[draw=none] (R1a) at (7.5,-1) {$ $};
	\node[draw=none] (R1b) at (15.5,-1) {$ $};
	\draw[line width=0.5mm] (R1a) -- (R1b);

\node at (6.5,-1) {$R_{2}$};
	\node[draw=none] (R11a) at (4.5,-1.5) {$ $};
	\node[draw=none] (R11b) at (8.5,-1.5) {$ $};
	\draw[line width=0.5mm] (R11a) -- (R11b);

    \node at (5.5,-2.5) {$R'_{2}$};
	\node[draw=none] (R11a) at (3.5,-2) {$ $};
	\node[draw=none] (R11b) at (7.5,-2) {$ $};
	\draw[line width=0.5mm] (R11a) -- (R11b);

	\node[draw=none] (RLka) at (-1,-1.25) {$ $};
	\node[draw=none] (RLkb) at (0.5,-1.25) {$ $};
	\draw[line width=0.5mm] (RLka) -- (RLkb);
	\node at (-0.25,-0.75) {$M$};

    \node[draw=none] (Lka) at (-3.2,-2) {$ $};
	\node[draw=none] (Lkb) at (-0.2,-2) {$ $};
	\draw[line width=0.5mm] (Lka) -- (Lkb);
	\node at (-1.75,-2.5) {$L'_{k-1}$};
    \node[draw=none] (RLka) at (-0.3,-2) {$ $};
	\node[draw=none] (RLkb) at (2.7,-2) {$ $};
	\draw[line width=0.5mm] (RLka) -- (RLkb);
	\node at (1.25,-2.5) {$R'_{k-1}$};
	
	\node[draw=none] (Lka) at (-3.5,-1) {$ $};
	\node[draw=none] (Lkb) at (-0.5,-1) {$ $};
	\node at (-2,-0.5) {$L_{k-1}$};
	\draw[line width=0.5mm] (Lka) -- (Lkb);

	\node[draw=none] (Lka) at (0,-1) {$ $};
	\node[draw=none] (Lkb) at (3,-1) {$ $};
	\node at (1.5,-0.5) {$R_{k-1}$};
	\draw[line width=0.5mm] (Lka) -- (Lkb);

	
	\end{tikzpicture} 
	\caption{Instance where algorithm by Garay et al.\cite{garay1997efficient} gets CR $(2k-2)$.}\label{fig:garay-2k}
\end{figure}
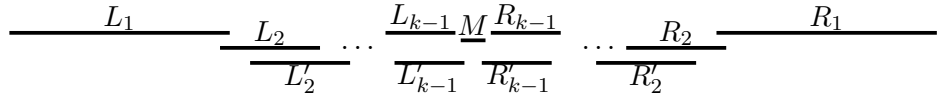
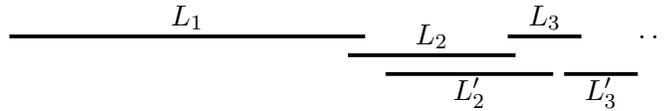
\begin{figure}[H]
	\centering
	
	\begin{tikzpicture}[scale=0.5]
	
	\node at (-12,-0.5) {$L_{1}$};
	\node[draw=none] (L1a) at (-17,-1) {$ $};
	\node[draw=none] (L1b) at (-7,-1) {$ $};
	\draw[line width=0.5mm] (L1a) -- (L1b);

 \node at (-5.5,-1) {$L_{2}$};
	\node[draw=none] (L1a) at (-8,-1.5) {$ $};
	\node[draw=none] (L1b) at (-3,-1.5) {$ $};
	\draw[line width=0.5mm] (L1a) -- (L1b);

    \node at (-4.5,-2.5) {$L'_{2}$};
	\node[draw=none] (L11a) at (-7,-2) {$ $};
	\node[draw=none] (L11b) at (-2,-2) {$ $};
	\draw[line width=0.5mm] (L11a) -- (L11b);

 \node at (-2.5,-0.5) {$L_{3}$};
	\node[draw=none] (L3a) at (-3.75,-1) {$ $};
	\node[draw=none] (L3b) at (-1.25,-1) {$ $};
	\draw[line width=0.5mm] (L3a) -- (L3b);

    \node at (-1,-2.5) {$L'_{3}$};
	\node[draw=none] (L33a) at (-2.25,-2) {$ $};
	\node[draw=none] (L33b) at (0.25,-2) {$ $};
	\draw[line width=0.5mm] (L33a) -- (L33b);

  \node[draw=none] (d1) at (0.5,-1) {$\dots$};

	
	\end{tikzpicture} 
	\caption{Single chain from instance of figure \ref{fig:garay-2k}}\label{fig:garay-2k-helper}
\end{figure}
\begin{observation}
    There exists an instance where algorithm \ref{alg:SUB} achieves a competitive ratio of $2k$, whereas the algorithm by Garay et al. gets 1.
\end{observation}
This instance \footnote{This instance is taken from lecture notes by Yossi Azar \url{http://www.cs.tau.ac.il/~azar/Online-Class10.pdf} .} is shown in figure \ref{fig:azar-fig}. Our algorithm ends up with a single interval in its solution, whereas |OPT| = $2k$.
\begin{figure}[H]
	\centering
	
	\begin{tikzpicture}[scale=0.8]

	\node at (-5,0.5) {$L_{1}$};
	\node[draw=none] (I1a) at (-10,0) {$ $};
	\node[draw=none] (I1b) at (0,0) {$ $};
	\draw[line width=0.5mm] (I1a) -- (I1b);

	\node at (-13,1) {$t$};
 \draw [-to](-12.5,0) -- (-12.5,2);
	
	\node[draw=none] (I2a) at (-1,0.5) {$ $};
	\node[draw=none] (I2b) at (1,0.5) {$ $};
	\node at (0,1) {$R_{2}$};
	\draw[line width=0.5mm] (I2a) -- (I2b);

	\node[draw=none] (I3a) at (-11,0.5) {$ $};
	\node[draw=none] (I3b) at (-9,0.5) {$ $};
	\node at (-10,1) {$R_{1}$};
	\draw[line width=0.5mm] (I3a) -- (I3b);
\node at (-5,2) {$L_{2}$};
	\node[draw=none] (L2a) at (-7,1.5) {$ $};
	\node[draw=none] (L2b) at (-3,1.5) {$ $};
	\draw[line width=0.5mm] (L2a) -- (L2b);

\node[draw=none] (R4a) at (-4,2) {$ $};
	\node[draw=none] (R4b) at (-2,2) {$ $};
	\node at (-3,2.5) {$R_{4}$};
	\draw[line width=0.5mm] (R4a) -- (R4b);

\node[draw=none] (R3a) at (-8,2) {$ $};
	\node[draw=none] (R3b) at (-6,2) {$ $};
	\node at (-7,2.5) {$R_{3}$};
	\draw[line width=0.5mm] (R3a) -- (R3b);
 
\node[draw=none] (d5) at (-5,3.5) {$\rvdots$};
 
	\end{tikzpicture} 
	\caption{Instance where our algorithm gets competitive ratio $2k$, whereas Garay et al. get 1.}\label{fig:azar-fig}
\end{figure}
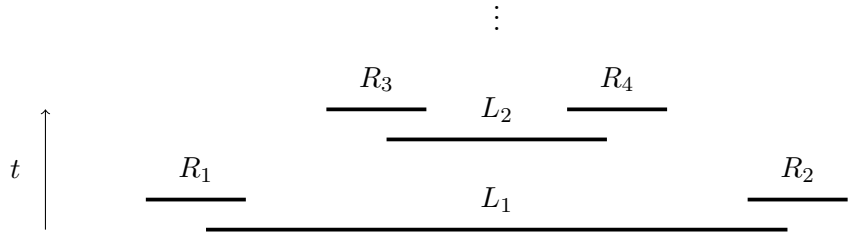

\section{Arbitrary Weights}
The case of intervals having an arbitrary weights has previously been considered for the case of single-length instances and ordered arrivals. Woeginger \cite{woeginger1994line} gives an optimal deterministic algorithm that is 4-competitive. Fung et al. \cite{fung2014improved} give a barely random algorithm that is 2-competitive, and show that it is optimal amongst barely random algorithms that choose between two deterministic algorithms. Woeginger \cite{woeginger1994line} shows that in the case of two different lengths, there does not exist a deterministic algorithm with finite competitive ratio. We show how to combine the barely random algorithm of Fung et al., with a classify and randomly select algorithm, to obtain a randomized algorithm for the any-order case, that achieves a competitive ratio of $2k$, when there are $k$ different lengths.\\\\
First, one can observe that the 2-competitive single-length algorithm by Fung et al. \cite{fung2014improved} (Theorem 3.1), works even in the case of any-order arrivals. Our algorithm (denoted as $ARB$), which requires knowledge of all the different lengths of the instance, works as follows: Choose one of $k$ lengths, uniformly at random. Then execute the algorithm of Fung et al., looking only at intervals of the chosen length.

\begin{theorem}
    Algorithm $ARB$ achieves a competitive ratio of $2k$, for the problem of any-order interval selection, with $k$ different lengths and arbitrary weights.
\end{theorem}
\begin{proof}
    Let $L_{1},L_{2},...,L_{k}$ be all the different lengths of an instance. Associated with length $L_{i}$, is a sub-instance $C_{i}$, comprised only of the intervals of length $L_{i}$. Let $OPT_{i}$ denote the weight of an optimal solution on sub-instance $C_{i}$. The expected performance of the algorithm can be bounded as follows:
    $$ \E[ALG] \geq \frac{1}{k}\frac{OPT_{1}}{2} + \frac{1}{k}\frac{OPT_{2}}{2} +...+\frac{1}{k}\frac{OPT_{k}}{2} \geq \frac{OPT}{2k}$$
    The first inequality holds because applying Fung et al. \cite{fung2014improved} on $C_{i}$ gives a solution of weight at least $\frac{OPT_{i}}{2}$. The second inequality holds because for every length $L_{j}$, the total weight of the intervals of length $L_{j}$ in the final solution, is at most $OPT_{j}$.
\end{proof}
We note that the algorithm does not need to know the actual lengths beforehand, or even $k$. The algorithm can start working with the first length that appears. When a second length arrives, the algorithm discards its current solution and chooses the new length with probability $\frac{1}{2}$. More generally, when the $i$th length arrives, the algorithm starts over using the new length with probability $\frac{1}{i}$. One can see that the probability that any length is chosen is $\frac{1}{k}$. Moreover, by replacing the 2-competitive arbitrary weights algorithm with a simple greedy algorithm, we get a randomized algorithm for the unweighted case that is $2k$-competitive and does not use revoking (as long as we know $k$).

\section{Random Order}
In this section, we assume the adversary chooses the instance configuration, but the intervals arrive in a random order. We consider unweighted, single-lengthed instances, and deterministic memoryless algorithms with revocable acceptances.
We consider various cases and show that the only type of algorithm that can possibly benefit from the random order model is a \textit{one-directional} algorithm, namely an algorithm that only replaces intervals on the left side, or only on the right side, regardless of the amount of overlap. For any other algorithm, we show how the adversary can enforce a competitive ratio of 2, resulting in no benefit over adversarial arrivals for single-lengthed instances. \\\\ 

 On the instances we present, an algorithm only keeps one interval in its solution at any given time w.h.p., so the decision on taking or discarding a newly arrived interval, depends only on the local conflicts. The behavior of an algorithm is described by two functions, $F_{l}$ and $F_{r}$: $F_{l}(v) \in \{0,1\}$ denotes whether the algorithm replaces $I_{old}$ with $I_{new}$ when the conflict is on the left of $I_{old}$, and the overlap is equal to $v$. $F_{r}$ is defined similarly for conflicts on the right. We are concerned with single-lengthed instances, where there can only be partial conflicts. We also assume that a new interval is never taken if it conflicts with more than one existing intervals in the solution, and show why such an action cannot benefit the algorithm.\\\\
One might notice that in the above description, information about the endpoints of the conflicting intervals is omitted. This was to improve readability, and we do in fact allow a deterministic memoryless algorithm to know the endpoints of intervals. The lower bounds presented in this section still hold, regardless of where the intervals are placed on the line.\\\\

Let $L$ denote the interval length of an instance.\\
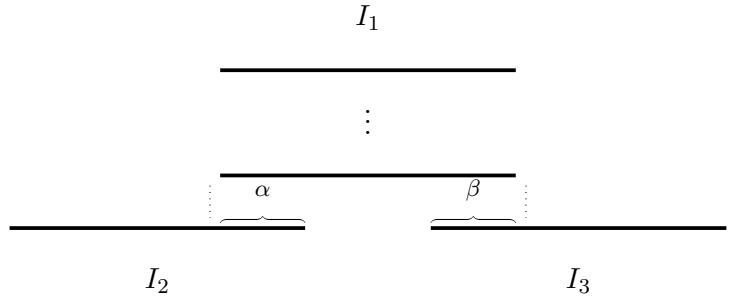
\begin{figure}[H]
	\centering
	
	\begin{tikzpicture}[scale=0.7]

	\node at (-3,1) {$I_{1}$};
	\node[draw=none] (I1a) at (-6,0) {$ $};
	\node[draw=none] (I1b) at (0,0) {$ $};
	\draw[line width=0.5mm] (I1a) -- (I1b);
	
	\node[draw=none] (d1) at (-3,-1) {$\rvdots$};
	
	\node[draw=none] (I11a) at (-6,-2) {$ $};
	\node[draw=none] (I11b) at (0,-2) {$ $};
	\draw[line width=0.5mm] (I11a) -- (I11b);

	\node[draw=none] (I2a) at (-10,-3) {$ $};
	\node[draw=none] (I2b) at (-4,-3) {$ $};
	\node at (-7,-4) {$I_{2}$};
	\draw[line width=0.5mm] (I2a) -- (I2b);

	\node[draw=none] (I3a) at (-2,-3) {$ $};
	\node[draw=none] (I3b) at (4,-3) {$ $};
	\node at (1,-4) {$I_{3}$};
	\draw[line width=0.5mm] (I3a) -- (I3b);
	
	\node[draw=none] (I2c) at (-6,-3) {$ $};
	\node[draw=none] (I3c) at (0,-3) {$ $};
	
	\draw[dotted] (I11a) -- (I2c) ;
	\draw [decorate,decoration={calligraphic brace,raise=2pt}]
	(I2c) -- (I2b) node [black,midway,yshift=0.5 cm]
	{\footnotesize $\alpha$};
	
	\draw[dotted] (I11b) -- (I3c) ;
	\draw [decorate,decoration={calligraphic brace,raise=2pt}]
	(I3a) -- (I3c) node [black,midway,yshift=0.5 cm]
	{\footnotesize $\beta$};

	\end{tikzpicture} 
	\caption{Random order bad instance. $\alpha < \frac{L}{2}$ and $\beta < \frac{L}{2}$}\label{fig:rom_inst}
\end{figure}
Figure \ref{fig:rom_inst} depicts the general structure of the main bad instance in the random order model. There is a single copy of intervals $I_{2}$ and $I_{3}$, but a very large number of identical intervals $I_{1}$. We first prove the following lemmas about two different algorithms. We refer to an algorithm as \textit{always-replace}, if the new interval is always taken whenever a conflict occurs. Respectively, a \textit{never-replace} algorithm never takes the new interval when there's a conflict.

\begin{lemma}
An \textit{always-replace} algorithm has a competitive ratio of 2 for the instance of figure \ref{fig:rom_inst}.
\end{lemma}
\begin{proof}
Because there is a large number of $I_{1}$ intervals, when we look at the arrival sequence of intervals, $I_{1}$ will both precede and follow the arrival of $I_{2}$ and $I_{3}$ with very high probability. This will result in the algorithm ending up with a single interval $I_{1}$ in its solution, whereas the size of an optimal solution is 2.
\end{proof}

\begin{lemma}
A \textit{never-replace} algorithm has a competitive ratio of 2 for the instance of figure \ref{fig:rom_inst}.
\end{lemma}
\begin{proof}
The first online interval will be $I_{1}$ w.h.p. It will never be replaced and the algorithm will end up with one interval in its solution, admitting a competitive ratio of 2.
\end{proof}

\subsection{Overlap at most $\frac{L}{2}$}

We first consider the behavior of an algorithm for overlaps at most half the length of the interval. If there is an overlap amount $v\leq \frac{L}{2}$, such that $F_{l}(v) = F_{r}(v)$, then the adversary can use the aforementioned instance with $\alpha = \beta = v$. The algorithm's behavior would then be either that of an \textit{always-replace} algorithm, or that of a \textit{never-replace} algorithm, incurring a competitive ratio of 2. For an algorithm to do better, it must be that $F_{l}(v) \neq F_{r}(v), \forall v\leq \frac{L}{2}$. In other words, for each overlap, the algorithm would replace in one way. Assume now that there exist two different overlap amounts, $\alpha, \beta \leq \frac{L}{2}$, that replace in different directions, namely $F_{l}(\alpha)\neq F_{l}(\beta)$ and $F_{r}(\alpha)\neq F_{r}(\beta)$. W.l.o.g. assume $F_{r}(\alpha) = 1$. In this case, the adversary can use the instance of figure \ref{fig:rom_inst}. Interval $I_{1}$ arrives first w.h.p., and intervals $I_{2}$ and $I_{3}$ are rejected whenever they arrive, resulting in 2-competitiveness. To avoid this, it must be that $F_{l}(\alpha)= F_{l}(\beta)$ and $F_{r}(\alpha)= F_{r}(\beta)$ $\forall \alpha, \beta \leq \frac{L}{2}$, meaning that for overlaps at most half the length of the interval, the algorithm replaces in one direction.

\subsection{Overlap greater than $\frac{L}{2}$}

We now consider the algorithm's behavior for overlap amounts greater than $\frac{L}{2}$, knowing that for overlap $v\leq \frac{L}{2}$, the algorithm is one-directional. W.l.o.g. we assume that $F_{l}(v) = 1, \forall v\leq \frac{L}{2}$.

\begin{lemma}
If $\exists\, \gamma > \frac{L}{2}: F_{l}(\gamma) = 0$, the adversary can force a competitive ratio of 2.
\end{lemma}
\begin{proof}
Consider the instance of figure \ref{fig:rom-bad-2}. Because of the multiple copies of $I_{1}$, the arrival of $I_{2}$ and $I_{3}$ is preceded by $I_{1}$ intervals w.h.p. We know that $F_{r}(\alpha)=0$, and given that $F_{l}(\gamma) = 0$, we have that $I_{2}$ and $I_{3}$ will be rejected on arrival, leaving the algorithm with a single interval in its solution.
\end{proof} 
\begin{lemma}
If $\exists\,\gamma > \frac{L}{2}: F_{r}(\gamma) = 1$, the adversary can force a competitive ratio of 2.
\end{lemma}
\begin{proof}
Using the previous lemma, we assume that $F_{l}(\gamma) = 1$. Using the same instance of figure \ref{fig:rom-bad-2}, w.h.p. $I_{2}$ will conflict with, and replace $I_{1}$ on arrival, and will then be replaced by another arrival of $I_{1}$. Interval $I_{3}$ will again be rejected on arrival, leaving the algorithm with $I_{1}$ in its solution.
\end{proof} 
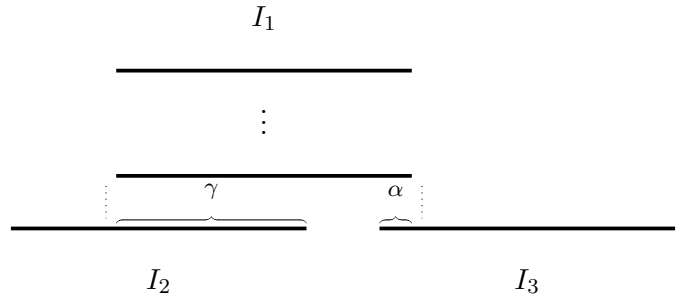
\begin{figure}[H]
	\centering
	
	\begin{tikzpicture}[scale=0.7]

	\node at (-3,1) {$I_{1}$};
	\node[draw=none] (I1a) at (-6,0) {$ $};
	\node[draw=none] (I1b) at (0,0) {$ $};
	\draw[line width=0.5mm] (I1a) -- (I1b);
	
	\node[draw=none] (d1) at (-3,-1) {$\rvdots$};
	
	\node[draw=none] (I11a) at (-6,-2) {$ $};
	\node[draw=none] (I11b) at (0,-2) {$ $};
	\draw[line width=0.5mm] (I11a) -- (I11b);

	\node[draw=none] (I2a) at (-8,-3) {$ $};
	\node[draw=none] (I2b) at (-2,-3) {$ $};
	\node at (-5,-4) {$I_{2}$};
	\draw[line width=0.5mm] (I2a) -- (I2b);

	\node[draw=none] (I3a) at (-1,-3) {$ $};
	\node[draw=none] (I3b) at (5,-3) {$ $};
	\node at (2,-4) {$I_{3}$};
	\draw[line width=0.5mm] (I3a) -- (I3b);
	
	\node[draw=none] (I2c) at (-6,-3) {$ $};
	\node[draw=none] (I3c) at (0,-3) {$ $};
	
	\draw[dotted] (I11a) -- (I2c) ;
	\draw [decorate,decoration={calligraphic brace,raise=2pt}]
	(I2c) -- (I2b) node [black,midway,yshift=0.5 cm]
	{\footnotesize $\gamma$};
	
	\draw[dotted] (I11b) -- (I3c) ;
	\draw [decorate,decoration={calligraphic brace,raise=2pt}]
	(I3a) -- (I3c) node [black,midway,yshift=0.5 cm]
	{\footnotesize $\alpha$};

	\end{tikzpicture} 
	\caption{Random order bad instance. $\alpha < \frac{L}{2}$ and $\gamma > \frac{L}{2}$}\label{fig:rom-bad-2}
\end{figure}
We now explain why an algorithm doesn't gain anything by replacing more than one interval at a time.
First, notice that in all the negative results presented in this section, there's never a conflict between more than two intervals w.h.p. Because we consider single-lengthed instances, a newly arrived interval can conflict with at most two other intervals currently in the algorithm's solution. Assume that for some overlap amounts $\alpha$ and $\beta$, a memoryless algorithm accepts a new interval conflicting with two already accepted intervals. At least one of $\{\alpha , \beta\}$ must be less than $\frac{L}{2}$ (otherwise the two current intervals would have a conflict, a contradiction). The adversary can then use the instance from figure \ref{fig:rom_inst} or figure \ref{fig:rom-bad-2}, with the appropriate $\alpha$ and $\beta$ overlaps. Regardless of the arrival order of $\{I_{2},I_{3}\}$, the algorithm ends up with a single interval in its solution.\\\\
Combining sections 5.1 and 5.2, we get the following theorem:
\begin{theorem}
Every deterministic memoryless algorithm that isn't one-directional, can be forced to a competitive ratio of at least 2 for the problem of online unweighted single-lengthed interval selection under random order arrivals.
\end{theorem}

\section{Conclusions \& Open Problems}
There are a number of possible directions for future work. A very natural direction is looking at specific weighted cases. Deterministically, Garay et al. \cite{garay1997efficient} have settled the case of proportional weights with an optimal, constant-competitive algorithm. It's interesting to see if a similar constant can be achieved for the more general weight functions studied by Woeginger \cite{woeginger1994line}, with or without randomness. We considered the case of arbitrary weights in Section 4.\\\\
It is fair to say that we have a very limited understanding of randomized algorithms for interval selection.  In the unweighted adversarial setting, we have shown that no memoryless randomized algorithm can be constant-competitive and Fung et al. show that with one random bit, their 2-competitive algorithm is optimal. But we have no other negative results for unweighted or weighted interval selection when revoking is permitted. We would like to extend the memoryless model to algorithms with constant memory (beyond the current solution) as discussed further in the appendix. In particular, we would want to allow access to a few initial random bits which would also capture algorithms that fall under the \textit{classify and randomly select} paradigm. It would also be interesting to restrict the number of copies the adversary can generate, maybe only allowing a single copy of every interval, and see if memoryless randomized algorithms become more powerful.\\\\
As mentioned earlier, we can think of the parameter $k$ as a refinement of the total number of intervals, and the number of vertices of a call control instance. We find it interesting to see if restricting the number of different lengths can yield better results for the problem of call control on other classes of graphs, such as trees (see \cite{awerbuch1994competitive}).\\\\
Finally, to the best of our knowledge, we have initiated the study of this model under random order arrivals, where there are many open questions for future work. We have only looked at single-lengthed instances, a special case that, in the adversarial setting, doesn't even require revoking. Looking at multiple lengths under random arrivals, is a natural next step. Lastly, we have shown that one-directional algorithms for single-lengthed instances, are the only type of deterministic memoryless algorithms that can possibly benefit from random order arrivals. We don't have any provable upper bounds on the performance of a one-directional algorithm, but
we have conducted experiments that suggest it may achieve much better than 2-competitiveness. This is an interesting contrast with the adversarial model, where a one-directional algorithm would perform arbitrarily bad.\\\\
\textbf{Acknowledgements:} We would like to thank Denis Pankratov, Adi Ros{\'e}n and Omer Lev for many helpful comments.

\printbibliography
\appendix
\section{Memory in online computation}\label{app:A}
The impact of limited memory (or time) is usually not considered in online competitive analysis, since the analysis is information theoretic and independent of complexity issues. Of course, the assumption is that algorithms are usually efficient (in terms of time and space) while negative results are that much stronger as they do not require any complexity assumptions. 

However, the arguments for limited memory in streaming algorithms apply equally well to online algorithms which are forced to immediately make decisions for each input as it occurs.  There has been some limited results concerning memory with respect to competitiveness. Perhaps the first study of memoryless algorithms occurs in Kleinberg's \cite{Kleinberg94} study of balancing algorithms for 2-server algorithms where it is shown that the optimal competitive ratio cannot be achieved. The first issue is  to define memoryless and  bounded  memory algorithms in online computation? In the earlier conference version of Adler and Azar \cite{adler2003beating}, they ask ``is there a memoryless online algorithm for interval selection that achieves a constant competitive ratio''. In the journal version they reframe this question and ask ``is there a bounded memory algorithm achieving a constant competitive ratio''. They do no provide a definite meaning for the term {\it bounded memory}.  
Emek et al. \cite{emek2016space} provide an interesting streaming based  online algorithm (with revoking) that is ``barely random'' and achieves an improved  constant competitive ratio. They seem to implicitly  argue that their algorithm is ``bounded memory''  in the sense that the additional memory (beyond the current solution) is linear in the size of the optimal solution.   
Here we are counting memory in terms of the number of intervals and not necessarily in terms of bits of memory. This is a ``permissive'' definition  of memoryless that could nicely serve in defining ``semi-streaming'' that goes beyond graph optimization problems.
In this ``semi-streaming'' model, Cabello and Pérez-Lantero \cite{cabello2017interval} give alternative algorithms that match the performance of Emek et al. \cite{emek2016space} for interval selection  and same-length interval selection. In addition, for interval selection (on $n$ equally spaced points) they show how to ($2+\epsilon$)-approximate the optimal solution \textit{size} using $O(\epsilon^{-5}log^{6}n)$ space, and show that no better approximation can be achieved using $o(n)$ space. This lower bound on memory also applies to algorithms for computing a solution in the model proposed by Halld{\'o}rsson et al. \cite{halldorsson2010streaming}.\\\\
In section 3, we define memoryless in a strict sense, namely that the algorithm does not maintain any information except the current solution. This is the definition  of memoryless as used in Raghavan and Snir \cite{raghavan1989memory}, Koutsoupias \cite{koutsoupias2009k}, and Coester and Koutsoupias \cite{coester2019online}. The strict definition is also sufficient for the simple $\frac{1}{4}$ competitive,  1 random bit randomized algorithm (without revoking) for the proportional knapsack\footnote{In the proportional knapsack the profit of an item is equal to its size. We assume every item has $size > 0$ or that the representation of the $i^{th}$ item includes the index $i$.}.
But as stated the strict definition does not include barely random algorithms even for those which do not use any memory beyond remembering a few initial bits. However, we would argue that remembering any initial random bits is a form of memory. \\\\
This leads us to what is arguably the most interesting interpretation of the Adler and Azar question; namely, is there a constant competitive (perhaps barely random) randomized algorithm that does not store any information besides the current solution and some number of initial bits. Of course we would allow such algorithms to use fresh bits (as well as the current solution) in randomly deciding the decision for the current input item.   The Emek et al. algorithm uses memory well beyond the 2 random bits. In contrast, the Fung et al. \cite{fung2014improved} algorithm for single length, arbitrary weights, does not use any additional memory beyond the one initial random bit. Indeed, 
many classify and randomly select algorithms only remember the initial random bits needed to classify an input item. We can also ask more generally when an algorithm only maintains a constant number of bits (and not necessarily initial random bits) in both deterministic and randomized algorithms. B{\"{o}}ckenhauer et al. \cite{bockenhauer2014online} show  that the proportional and general knapsack problems exhibit ``phase transitions''  as to how much advice and random bits are needed to achieve certain certain competitive ratios. The results of \cite{pena2019extensions,BuchbinderNW23,BockenhauerKKK2017,DurrKR2016} provide interesting phase transitions for randomized algorithms and deterministic algorithms with advice for the online unweighted  bipartite matching problem. \\\\
Finally we mention that the results of Mikklesen \cite{mikkelsen2015randomization} (for repeatable problems) and B{\"{o}}ckenhauer \cite{BockenhauerKKK2017} provide interesting results about the relation between randomized algorithms and advice. The bipartite matching problem and the interval selection problem are repeatable problems. It is interesting to explore this relation further with regard to the interval selection problem. Namely, does interval selection have a phase transition in that 1 bit of randomness is sufficient for a barely random $2$ competitive ratio whereas no additional random or advice bits can help, or is there perhaps some threshold at $\Theta({\log n})$ where that amount of  advice (randomness) can asymptotically  beat the $2$ ratio for interval selection. 
It is interesting to observe the difference between advice  bits and random bits for the proportional and general knapsack problems as proven in B{\"{o}}ckenhauer et al. \cite{bockenhauer2014online}. Is there a provable difference for interval selection between advice bits and random bits?

\end{document}